\documentclass{llncs}
\usepackage{graphicx}
\usepackage{url}
\usepackage{amsmath}
\usepackage{amssymb}
\usepackage{amsfonts}
\usepackage{subfigure}
\usepackage{enumerate}
\usepackage{paralist}
\usepackage[textsize=footnotesize]{todonotes}

\usepackage{rotating}
\usepackage{amsthm}
\usepackage{times}
\usepackage{enumitem}
\usepackage{stmaryrd}
\usepackage{wrapfig}

\graphicspath{{fig/}}

\makeatletter
\def\cramped                           
 {
\parskip0pt\@topsep0pt       
  \itemsep0pt\parsep0pt
  \settowidth{\labelwidth}{\@itemlabel}
  \advance\leftmargin-\labelsep
  \advance\leftmargin-\labelwidth
  \parshape1\@totalleftmargin\linewidth
}
\makeatother

\newcommand{\wleft}{\mathit{left}}
\newcommand{\wright}{\mathit{right}}
\newcommand{\LSVR}{$\Phi$-LSVR}
\newcommand{\phit}{$\Phi$-transversal}
\newcommand{\SA}{\includegraphics{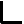}~}
\newcommand{\SB}{\includegraphics{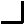}~}
\newcommand{\SC}{\includegraphics{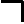}~}
\newcommand{\SD}{\includegraphics{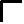}~}

\newcommand{\face}[2]{f^{#1}_{#2}}

\newcommand{\rface}[1]{\face{}{r}(#1)}
\newcommand{\bface}[1]{\face{}{b}(#1)}

\newcommand{\algo}{\texttt{$\Phi$LSVRDrawer}}

\title{Simultaneous Visibility Representations \\ of Plane $st$-graphs
  Using L-shapes \thanks{Research supported in part by the MIUR
    project AMANDA ``Algorithmics for MAssive and Networked DAta'',
    prot. 2012C4E3KT\_001 and NSERC Canada.}}
\author{
William S. Evans\inst{1} \and
Giuseppe Liotta\inst{2} \and
Fabrizio Montecchiani\inst{2}
}
\institute{
University of British Columbia, Canada\\
\email{will@cs.ubc.ca}
\and
Universit{\`a} degli Studi di Perugia, Italy\\
\email{\{giuseppe.liotta,fabrizio.montecchiani\}@unipg.it}
}

\begin{document}

\pagenumbering{arabic}
\pagestyle{plain}

\maketitle

\begin{abstract}
Let $\langle G_r,G_b \rangle$ be a pair of plane $st$-graphs with the
same vertex set $V$.
A simultaneous visibility representation with 
L-shapes of $\langle G_r,G_b \rangle$ is a pair
of bar visibility representations $\langle\Gamma_r,\Gamma_b\rangle$ such
that, for every vertex $v \in V$, $\Gamma_r(v)$ and $\Gamma_b(v)$ are 
a horizontal and a vertical segment, which share an end-point. 
In other words, every vertex is drawn as an $L$-shape,  
 every edge of $G_r$ is a vertical visibility segment, and 
every edge of $G_b$ is a horizontal visibility segment. 
Also, no two L-shapes intersect each other. 
An L-shape has four possible rotations, 
and we assume that each vertex is given a 
rotation for its $L$-shape as part of the input. Our main results are:
(i) a characterization of those pairs of plane $st$-graphs admitting
such a representation, (ii) a cubic time
algorithm to recognize them, and (iii) a linear time drawing
algorithm if the test is positive. 
\end{abstract}

\section{Introduction}

Let $G_r$ and $G_b$ be two plane graphs with the same vertex set.
A {\em simultaneous embedding (SE)} of $\langle G_r,G_b \rangle$ consists of two 
planar drawings, $\Gamma_r$ of $G_r$ and $\Gamma_b$ of
$G_b$, such that  every edge is a simple Jordan arc, and 
every vertex is the same point both in $\Gamma_r$ and in $\Gamma_b$.
The problem of computing SEs has received a lot
of attention in the Graph Drawing literature,  
partly for  its theoretical interest and partly for its application
to the visual analysis of dynamically changing networks on a common
(sub)set of vertices. For example, it is known that any two plane
graphs with the same vertex set admit a SE
where the edges are polylines with at most two bends, which are sometimes 
necessary~\cite{DBLP:journals/ijcga/GiacomoL07}. If the edges are straight-line
segments, the representation is called a \emph{simultaneous
  geometric embedding (SGE)},
and many graph pairs do not have an 
SGE: a tree and a
path~\cite{DBLP:journals/jgaa/AngeliniGKN12}, a planar graph and a matching~\cite{DBLP:journals/jgaa/CabelloKLMSV11},
and three paths~\cite{DBLP:journals/comgeo/BrassCDEEIKLM07}.
On the positive side, the discovery of graph
 pairs that have an SGE is still a fertile research topic.
The reader can refer to the survey by Bl\"asius, Kobourov and Rutter~\cite{Handbook} for references and open problems.

Only a few papers study simultaneous representations that adopt a drawing
paradigm different from SE and SGE. A seminal paper by Jampani and
Lubiw initiates the study of {\em simultaneous intersection
representations (SIR)}~\cite{DBLP:journals/jgaa/JampaniL12}. In an intersection representation of a graph, 
each vertex is a geometric object and there is an edge
between two vertices if and only if the corresponding objects intersect. 
Let $\langle G_r,G_b \rangle$  be two graphs that have a
subgraph in common. 
A SIR of $\langle G_r,G_b \rangle$ is a pair of
intersection representations where each vertex in $G_r \cap G_b$ is
mapped to the same object in both realizations. 
Polynomial-time algorithms for testing the existence of SIRs for chordal, comparability, interval, and permutation graphs have been presented~\cite{DBLP:conf/soda/BlasiusR13,DBLP:conf/isaac/JampaniL10,DBLP:journals/jgaa/JampaniL12}. 

We introduce and study a different type of
simultaneous representation, where each graph is realized as a \emph{bar visibility representation} and two
segments representing the same vertex share an end-point. 
A bar visibility representation of a plane graph $G$ is an embedding preserving 
drawing $\Gamma$ where the vertices of $G$ are non-overlapping horizontal segments, 
and two segments are joined by a vertical visibility segment if and only if there exists an edge in $G$ between
the two corresponding vertices (see,
e.g.,~\cite{DBLP:journals/dcg/RosenstiehlT86,TamassiaTollis86}).
A visibility segment has thickness $\epsilon >0$ and does not intersect any other segment.


A {\em simultaneous visibility
representation with L-shapes} of $\langle G_r,G_b \rangle$ is a pair
of bar visibility representations $\langle\Gamma_r,\Gamma_b\rangle$ such
that for every vertex $v \in V$,
$\Gamma_r(v)$ and $\Gamma_b(v)$ are a horizontal and a vertical segment that share an end-point. 
In other words, every vertex is an L-shape, and 
every edge of $G_r$ (resp., $G_b$) is a vertical (resp., horizontal) visibility segment.
Also, no two L-shapes intersect. 
A simultaneous visibility representation with $L$-shapes of $\langle G_r,G_b \rangle$  
where the rotation of the $L$-shape of each vertex in $V$ is defined by a function
$\Phi:V \rightarrow \mathcal H=\{\SA,\SB,\SC,\SD\}$, is called a \emph{\LSVR{}} in the following.
While this definition does not assume any particular direction on
the edges of $G_r$ (resp., $G_b$), the resulting representation does induce
a bottom-to-top (resp., left-to-right) $st$-orientation.
In this paper, we assume that $G_r$ and $G_b$ are directed and this
direction must be preserved in the visibility representation.
Also, the two graphs have been augmented with distinct (dummy) sources and sinks. 
More formally, $G_r$=$(V \cup \{s_r,t_r\}, E_r)$ and
$G_b$=$(V \cup \{s_b,t_b\}, E_b)$ are two plane $st$-graphs with 
sources $s_r$, $s_b$, and sinks $t_r$, $t_b$. 

In terms of readability, this kind of simultaneous representation has the following advantages: $(i)$ The edges are depicted as straight-line segments (as in SGE) and the edge-crossings are rectilinear; $(ii)$ The edges of the two graphs are easy to distinguish, since they consistently flow from bottom to top for one graph and from left to right for the other graph. Having rectilinear crossing edges is an important benefit in terms of readability, as shown in~\cite{DBLP:conf/apvis/HuangHE08}, which motivated a relevant amount of research on right-angle crossing (RAC) drawings, see~\cite{dl-cargd-12} for a survey.

Our main contribution is summarized by the following theorem.

\begin{theorem}\label{th:mainth}
Let $G_r$ and $G_b$ be two 
plane $st$-graphs defined on the same set
of $n$ vertices $V$ and with distinct sources and sinks.  Let $\Phi:V
\rightarrow \mathcal H=\{\SA,\SB,\SC,\SD \}$. There exists an
$O(n^3)$-time algorithm to test whether $\langle G_r, G_b \rangle$
admits a \LSVR{}. Also, in the positive case, a \LSVR{} can be computed in $O(n)$ time.
\end{theorem}


{\noindent}This result relates to previous studies on topological rectangle visibility graphs~\cite{StreinuWhitesides03} and transversal structures (see, e.g.,~\cite{f-rsrpg-13,DBLP:journals/dm/Fusy09,DBLP:journals/tcs/KantH97,DBLP:conf/soda/Schnyder90}). 
Also, starting from a \LSVR{} of $\langle G_r, G_b \rangle$, we
can compute a \emph{simultaneous RAC embedding} (SRE) of the two graphs with
at most two bends per edge, improving the general upper bound
by Bekos {\em et al.}~\cite{DBLP:conf/walcom/BekosDKW15} for those pairs of graphs
that can be directed and augmented to admit a \LSVR{}. As an application of this result, in Section~\ref{se:wheel+matching} 
we show  an alternative proof of another result by Bekos {\em et al.}
that a wheel graph and a matching admit an SRE with at most two bends for each edge 
of the wheel, and no bends for the matching edges~\cite{DBLP:conf/walcom/BekosDKW15}.

The proof of Theorem~\ref{th:mainth} is based on a characterization described in Section~\ref{se:characterization}, which allows for an
efficient testing algorithm presented in Section~\ref{se:test}. 
Section~\ref{se:preliminaries} contains preliminaries.  In Section~\ref{se:discussion} we conclude with a discussion on possible research directions that arise from our research.

\section{Preliminaries}\label{se:preliminaries}

A graph $G=(V,E)$ is \emph{simple}, if it contains neither loops nor multiple edges. 
We consider simple graphs, if not otherwise specified. 
A \emph{drawing} $\Gamma$ of $G$ maps each vertex of $V$ to a point of
the plane and each edge of $E$ to a Jordan arc between its two
end-points. We only consider \emph{simple drawings}, i.e., 
drawings such that the arcs representing two edges have at most one point in common, 
which is either a common end-vertex or a common interior point where the two arcs properly cross.  
A drawing is \emph{planar} if no two arcs representing two edges cross. 
A planar drawing subdivides the plane into topologically
connected regions, called \emph{faces}. The unbounded region is called
the \emph{outer face}. A \emph{planar embedding} of a graph is an
equivalence class of planar drawings that define the same set of
faces. A graph with a given planar embedding is a \emph{plane}
graph. For a non-planar drawing, we can still derive an embedding
considering that the boundary of a face may consist also of edge
segments between vertices and/or crossing points of edges. 
The unbounded region is still called the outer face.

A graph is \emph{biconnected} if it remains connected after removing
any one vertex. A directed graph (a digraph for short) is biconnected if its underlying undirected graph is biconnected.
The \emph{dual graph} $D$ of a plane graph $G$ is a plane multigraph
whose vertices are the faces of $G$ with an edge between two faces
if and only if they share an edge.
If $G$ is a digraph, $D$ is
also a digraph whose dual edge $e^*$ for a primal edge $e$ is
conventionally directed from the face, $\wleft_{G}(e)$, on the left of
$e$ to the face, $\wright_{G}(e)$, on the right of $e$.
Since we also use the opposite convention, we let
$D^{\shortrightarrow}$ (resp., $D^{\shortleftarrow}$) be the dual whose
edges cross the primal edges from left to right (resp., right to left).


A \emph{topological numbering} of a digraph is an assignment, $X$, of
numbers to its vertices such that $X(u) < X(v)$ for every edge
$(u,v)$.
A graph admits a topological numbering if and only if it is acyclic. 
An acyclic digraph with a single source $s$ and a single sink $t$ is
called an \emph{$st$-graph}.
In such a graph, for every vertex $v$, there exists a directed path
from $s$ to $t$ that contains $v$~\cite{TamassiaTollis86}.
A \emph{plane $st$-graph} is an $st$-graph that is planar and embedded
such that $s$ and $t$ are on the boundary of the outer face.
In any $st$-graph, the presence of the edge $(s,t)$ guarantees that
the graph is biconnected. 
In the following we consider $st$-graphs that contain the edge $(s,t)$, 
as otherwise it can be added without violating planarity.
Let $G$ be a plane $st$-graph, then for each vertex $v$ of $G$ the
incoming edges appear consecutively around $v$, and so do the outgoing
edges. Vertex $s$ only has outgoing edges, while vertex $t$
only has incoming edges.  
This is a particular transversal structure (see
Section~\ref{se:characterization}) known as a \emph{bipolar
  orientation}~\cite{DBLP:journals/dcg/RosenstiehlT86,TamassiaTollis86}.
Each face $f$ of $G$ is bounded by two directed paths with a common origin
and destination, called the \emph{left path} and \emph{right path} of
$f$. For all vertices $v$ and edges $e$ on the left (resp., right) path of
$f$, we let $\wright_{G}(v) = \wright_{G}(e) = f$ (resp., $\wleft_{G}(v) =
\wleft_{G}(e) = f$).

Tamassia and Tollis~\cite{TamassiaTollis86} proved the following lemma.

\begin{lemma}[\cite{TamassiaTollis86}]\label{le:tt}
Let $G$ be a 
plane $st$-graph and let $D^{\shortrightarrow}$ be its dual
graph. Let $u$ and $v$ be two vertices of $G$. Then exactly one of the
following four conditions holds:
\begin{enumerate}
\item $G$ has a path from $u$ to $v$.
\item $G$ has a path from $v$ to $u$.
\item $D^{\shortrightarrow}$ has a path from $\wright_G(u)$ to $\wleft_G(v)$.
\item $D^{\shortrightarrow}$ has a path from $\wright_G(v)$ to $\wright_G(u)$.
\end{enumerate}
\end{lemma}

\begin{wrapfigure}{r}{36mm}
\centering
\includegraphics{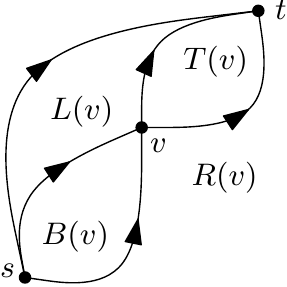}
\caption{Vertex sets $B(v)$, $T(v)$, $L(v)$, and $R(v)$ and their
  corresponding regions of the plane.}\label{fi:st-regions}
\end{wrapfigure}
Let $v$ be a vertex of $G$, then denote by $B(v)$ (resp., $T(v)$)
the set of vertices that can reach (resp., can be reached from)
$v$. Also, denote by $L(v)$ (resp., $R(v)$) the set of vertices
that are to the left (resp., to the right) of every path from
$s$ to $t$ through $v$. By Lemma~\ref{le:tt}, these four sets 
partition the vertices of $G \setminus \{v\}$.
In every planar drawing of $G$, they are
contained in four distinct regions of the plane that share point
$v$. The vertices of $B(v)$ are in the region delimited by the
leftmost and the rightmost paths from $s$ to $v$, while the vertices of $T(v)$ are in the region delimited by the
leftmost and the rightmost paths from $v$ to $t$. Edge $(s,t)$ separates the two
regions containing the vertices of $L(v)$ and $R(v)$, 
as in Fig.~\ref{fi:st-regions}. 
Refer to~\cite{dett-gd-99} for further details.

\newcommand{\reddual}{$D^{\shortrightarrow}_{r}$}
\newcommand{\bluedual}{$D^{\shortleftarrow}_{b}$}
\newcommand{\hreddual}{$\hat D^{\shortrightarrow}_{r}$}
\newcommand{\hbluedual}{$\hat D^{\shortleftarrow}_{b}$}

\section{Characterization}\label{se:characterization}


A \emph{transversal structure} of a plane graph $G$, is a coloring and an orientation of the inner edges (i.e., those edges that do not belong to the outer face) of the graph that obey some local and global conditions. Transversal structures have been widely studied and important applications  have been found. Bipolar orientations (also known as \emph{$st$-orientations}) of plane graphs have been used to compute bar visibility representations~\cite{DBLP:journals/dcg/RosenstiehlT86,TamassiaTollis86}. Further applications can be found in~\cite{f-rsrpg-13,DBLP:journals/dm/Fusy09,DBLP:journals/tcs/KantH97,DBLP:conf/soda/Schnyder90}. See also the survey by Eppstein~\cite{DBLP:conf/cccg/Eppstein10}.

To characterize those pairs of graphs that admit a \LSVR{}, we introduce a new transversal structure for the union of the two graphs (which may be non-planar) and show that it is in bijection with the desired representation. 
In what follows $G_r =(V_r = V \cup \{s_r,t_r\},E_r)$ and $G_b=(V_b = V \cup
\{s_b,t_b\},E_b)$ are two 
plane $st$-graphs with duals \reddual and \bluedual, respectively.

\begin{definition}\label{de:phit}
Given $\Phi: V \rightarrow \mathcal H=\{\SA,\SB,\SC,\SD\}$,
a \emph{(4-polar) \phit{}} is a drawing of a directed (multi)graph on the
vertex set $V \cup \{ s_r, t_r, s_b , t_b \}$ whose edges are
partitioned into red edges, blue edges, and the four special edges
$(s_r,s_b)$, $(s_b,t_r)$, $(t_r,t_b)$, and $(t_b,s_r)$ forming the
outer face, in clockwise order.
In addition, the \phit{} obeys the following conditions:
\begin{itemize}
\item[{\bf c1.}] The red (resp., blue) edges induce an $st$-graph with source $s_r$
(resp., $s_b$) and sink $t_r$ (resp., $t_b$).
\item[{\bf c2.}] For every vertex $u \in  V$, the clockwise order of
  the edges incident to $u$ forms four non-empty blocks of
  monochromatic edges, such that all edges in the same block are
  either all incoming or all outgoing with respect to $u$. The four blocks are
  encountered around $u$ depending on $\Phi(u)$ as in the following
  table.
  
  \medskip

\begin{tabular}{|c|c|c|c|}
\hline
&&&\\[-0.5em]
\includegraphics{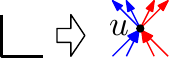} &
\includegraphics{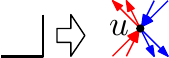} &
\includegraphics{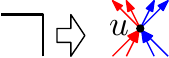} &
\includegraphics{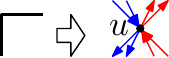} \\[0.5em]
\hline & & & \\[-0.5em]
 $\bface{u}{=}\wright_b(u)$ &
 $\bface{u}{=}\wright_b(u)$ &
 $\bface{u}{=}\wleft_b(u)$ &
 $\bface{u}{=}\wleft_b(u)$\\
 $\rface{u}{=}\wleft_r(u)$ &
 $\rface{u}{=}\wright_r(u)$ &
 $\rface{u}{=}\wright_r(u)$ &
 $\rface{u}{=}\wleft_r(u)$\\
\hline
\end{tabular}

\medskip

\item[{\bf c3.}] Only blue and red edges may cross and only if blue
  crosses red from left to right.
\end{itemize}
\end{definition}

A pair of plane $st$-graphs $\langle G_r, G_b \rangle$ \emph{admits} a
\phit{} if there exists a \phit{} $G_{rb}$ such
that restricting $G_{rb} \setminus \{ s_b, t_b \}$ to the red edges
realizes the planar embedding $G_r$ and
restricting $G_{rb} \setminus \{ s_r, t_r \}$ to the blue edges
realizes the planar embedding $G_b$.

Let $u$ be a vertex of $V$, then the edges of a single color enter and leave $u$ by the same face in the
embedding of the other colored graph. In other words, as condition {\bf c2} indicates, 
$\Phi(u)$ defines the face of $G_b$ (resp., $G_r$), denoted by $\bface{u}$ (resp., $\rface{u}$), 
by which the edges of $G_r$ (resp., $G_b$) incident to $u$ enter and leave $u$, in the \phit{}. 
Also, condition {\bf c3} implies
that edges
$\{(s_r,s_b)$, $(s_b,t_r)$, $(t_r,t_b)$, $(t_b,s_r)\}$
are not crossed, because they are not colored. 

\medskip
In the remainder of this section we will prove the next theorem.

\begin{theorem}\label{th:4polar}
Let $G_r$ and $G_b$ be two 
plane $st$-graphs defined on the same set of vertices $V$
and with distinct sources and sinks.
Let $\Phi:V \rightarrow \mathcal
H=\{\SA,\SB,\SC,\SD\}$. Then $\langle G_r, G_b \rangle$ admits a 
\LSVR{} if and only if it admits a \phit{}.
\end{theorem}

\begin{figure}[t]
\centering
\subfigure[]{\includegraphics{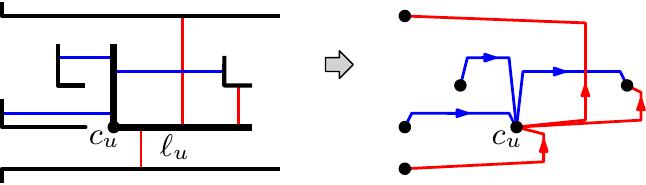}\label{fi:necessity}}\hfil
\subfigure[]{\includegraphics{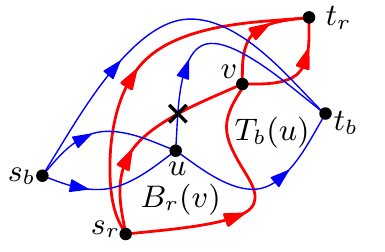}\label{fi:intersectingLshapes}}
\caption{(a) The replacement of the L-shape, $\ell_u$, for vertex $u$
  with its corner point $c_u$ and the drawing of $u$'s adjacent edges
  with $2$ bends per edge when
  constructing a \phit{} from a \LSVR{}.
 Only $\ell_u$'s visibilities are shown.
(b) Illustration for the proof of Lemma~\ref{le:lshapes}: the case when $u$ is in $B(v)$ and $v$ is in $T(u)$. }
\end{figure}
The necessity of the \phit{} is easily shown. Let
$\langle\Gamma_r,\Gamma_b\rangle$ be a \LSVR{} of $\langle G_r, G_b
\rangle$ with two additional  horizontal bars at the bottommost and
topmost sides of the drawing that represent $s_r$ and $t_r$, and two
additional vertical bars at the leftmost and rightmost sides of the
drawing that represent $s_b$ and $t_b$. From such a drawing we can
extract a \phit{} $G_{rb}$ as follows. Since the four vertices
$s_r$, $t_r$, $s_b$, and $t_b$ are represented by the extreme bars in
the drawing, these four vertices belong to the outer face, 
and the four edges on the outer face can be added without crossings. 
Also, we color red all inner edges represented by vertical visibilities
(directed from bottom to top), and blue all inner edges represented by
horizontal visibilities (directed from left to right).
To see that conditions {\bf c1}, {\bf c2} and {\bf c3} are satisfied, 
let $G_{rb}$ be a polyline drawing of computed as follows. 
Let $c_u$ be the corner of the L-shape, $\ell_u$, representing vertex
$u$.
For every edge $(u,v)$, replace its visibility segment by a polyline
from $c_u$ to $c_v$ that has two bends, both contained in the
visibility segment and each at distance $\delta$ from a different one of its
endpoints, for an arbitrarily small, fixed $\delta > 0$.
See Fig.~\ref{fi:necessity}.
Finally, replace every L-shape $\ell_u$ with its
corner $c_u$. Since each bar visibility representation preserves the embedding of the input graph, {\bf c1} is respected.
Also, {\bf c2} and {\bf c3} are clearly satisfied by the embedding derived from $G_{rb}$. 
We remark that, by construction, each edge is represented by a polyline with two bends
and two edges cross only at right angles; this observation will be used in Section~\ref{se:discussion}.

\medskip

To prove sufficiency, assume $\langle G_r, G_b \rangle$ admits a \phit{} $G_{rb}$. We present an
algorithm, \algo, that takes as input $G_{rb}$ and returns a \LSVR{} $\langle\Gamma_r,\Gamma_b\rangle$ of $\langle G_r, G_b \rangle$. We first recall the algorithm by Tamassia and Tollis (TT in the following) to compute an embedding preserving bar visibility representation of a 
plane $st$-graph $G$, see~\cite{dett-gd-99,TamassiaTollis86}: 
\begin{enumerate}
\item Compute the dual $D^{\shortrightarrow}$ of $G$.
\item Compute a pair of topological numberings $Y$ of $G$ and $X$ of $D^{\shortrightarrow}$. 
\item Draw each vertex $v$ as a horizontal bar with
$y$-coordinate $Y(v)$ and between $x$-coordinates $X(\wleft_{G}(v))$
and $X(\wright_{G}(v))-\epsilon$. 
\item Draw each edge $e=(u,v)$ as a vertical segment at $x$-coordinate $X(\wleft_{G}(e))$, between $y$-coordinates $Y(u)$ and $Y(v)$, and with thickness $\epsilon$.
\end{enumerate}

We are now ready to describe algorithm \algo.

\begin{description}
\item[Step 1:] Compute the dual graphs \reddual~of $G_r$ and \bluedual~of $G_b$. 


\item[Step 2:] Compute a pair of topological numberings $n_r$ of $G_r$ and $n_b$ of $G_b$.

\item[Step 3:] Compute a pair of topological numberings $n_r^*$ of \reddual~and $n_b^*$ of \bluedual.

\item[Step 4:] Compute a bar visibility representation $\Gamma_r$ of
  $G_r$ by using the TT algorithm with 
  $X(u)$=$X_r(u)$ = $n_r^*(u)$ and
  $Y(u)$=$Y_r(u)$=$n^*_b(\bface{u}) + n_r(u)\delta$, for each vertex $u$. 
  Also, shift the horizontal segment for each vertex $u$ to
  the left by $n_b(u)\delta$.

\item[Step 5:] Compute a bar visibility representation $\Gamma'_b$ of $G_b$ by 
  using the TT algorithm with 
  $X(u)$=$X_b(u)$=$n_b^*(u)$ and  $Y(u)$=$Y_b(u)$=$n^*_r(\rface{u}) + n_b(u)\delta$, for each vertex $u$. 
  Then turn $\Gamma'_b$ into a vertical bar visibility
  representation, $\Gamma_b$, by drawing every horizontal segment
  $((x_0,y),(x_1,y))$ in $\Gamma'_b$ as the vertical segment
  $((y,x_0),(y,x_1))$ in $\Gamma_b$. 
  Finally, shift the vertical segment for each vertex $u$ up by
  $n_r(u)\delta$.


\end{description}

Lemma~\ref{le:validstxy} guarantees that $Y_r$ and $Y_b$ are valid topological numberings, and thus, that $\Gamma_r$ and $\Gamma_b$ are two bar visibility representations. 
Also, Lemma~\ref{le:lshapes} ensures the union of $\Gamma_r$ and $\Gamma_b$ is a \LSVR{}. 
The shifts performed at the end of {\bf Steps 4-5} are to prevent the bars of two L-shapes from coinciding.  
The value $\delta > 0$ is chosen to be less than $\epsilon$ and less than the smallest difference between distinct numbers divided by the largest number from any topological numbering $n_r$, $n_b$, $n^*_r$, or $n^*_b$.
This choice of $\delta$ guarantees that all visibilities are preserved after the shift, 
and that no new visibilities are introduced.

\begin{lemma}\label{le:validstxy}
$Y_r$ is a valid topological numbering of $G_r$ and
$Y_b$ is a valid topological numbering of $G_b$.
\end{lemma}
\begin{proof}
Let $(u,v)$ be a red edge from $u$ to $v$.
We know that $n_r(u) < n_r(v)$.
Let $e_0$, $e_1$, $\dots$, $e_k$ be the blue edges crossed by $(u,v)$ in $G_{rb}$.
Due to conditions {\bf c2} and {\bf c3}, there exists a path
$\{
\bface{u} = \wright_b(e_0)$, 
$\wleft_b(e_0) = \wright_b(e_1)$, 
$\dots$, 
$\wleft_b(e_{k-1}) = \wright_b(e_k)$, 
$\wleft_b(e_k) = \bface{v} \}$ in \bluedual.
Thus, we also know that $n_b^*(\bface{u}) \leq n_b^*(\bface{v})$.
Since $Y_r(u) = n^*_b(\bface{u}) + n_r(u)\delta$ and $\delta >0$,
it follows that $Y_r(u) < Y_r(v)$.
A symmetric argument shows $Y_b(u) < Y_b(v)$ if $(u,v)$ is a blue edge.
\end{proof}

\begin{lemma}\label{le:lshapes}
Each vertex $u$ of $V$ is represented by an L-shape $\ell_u$ in $\langle \Gamma_r,\Gamma_b\rangle$ as defined by the function $\Phi$. Also no two L-shapes intersect each other.
\end{lemma}
\begin{proof}
Suppose $\Phi(u){=}~\SA$, as the other cases are similar.
Then, $\bface{u}{=}\wright_b(u)$ and $\rface{u}{=}\wleft_r(u)$.
The horizontal bar representing $u$ in $\Gamma_r$ is the segment 
$[p_0(u),p_1(u)]$, where the two points $p_0(u)$ and $p_1(u)$ are   
$p_0(u)= (n_r^*(\wleft_r(u)) + n_b(u)\delta$, $Y_r(u))$, and 
$p_1(u) = (n_r^*(\wright_r(u)) + n_b(u)\delta$, $Y_r(u))$. 
Note that, $n_r^*(\wleft_r(u)) < n_r^*(\wright_r(u))$.
The vertical bar representing $u$ in $\Gamma_b$ is the segment 
$[q_0(u),q_1(u)]$, where the two points $q_0(u)$ and $q_1(u)$ are  
$q_0(u)= (Y_b(u)$, $n_b^*(\wright_b(u)) + n_r(u)\delta)$, and
$q_1(u) = (Y_b(u)$, $n_b^*(\wleft_b(u)) + n_r(u)\delta)$. 
Note that, $n_b^*(\wright_b(u)) < n_b^*(\wleft_b(u))$.
Since $Y_r(u) = n^*_b(\bface{u}) + n_r(u)\delta = n^*_b(\wright_b(u))
+ n_r(u)\delta$, the bottom coordinate of the vertical bar
representing $u$ matches the $y$-coordinate of the horizontal bar
representing $u$.
Since $Y_b(u) = n^*_r(\rface{u}) + n_b(u)\delta = n^*_r(\wleft_r(u))
+ n_b(u)\delta$, the left coordinate of the horizontal bar
representing $u$ matches the $x$-coordinate of the vertical bar
representing $u$.
Thus the two bars form the L-shape \SA.

We now show that no two L-shapes properly intersect each other. 
Suppose by contradiction that the vertical bar of a vertex $u$, 
properly intersects the horizontal bar of a vertex $v$. 
Based on $\Phi$, the vertical bar of $u$ involved in the intersection is either a left
vertical bar or a right vertical bar, and it is drawn at
$x$-coordinate $n^*_r(\wleft_r(u)) + n_b(u)\delta$ or
$n^*_r(\wright_r(u)) + n_b(u)\delta$,
respectively.
Suppose it is a left vertical bar, as the other case is symmetric. 
Since $u$'s vertical bar properly intersects $v$'s horizontal bar,
we know by construction that
$
n^*_r(\wleft_r(v)) + n_b(v)\delta
<
n^*_r(\wleft_r(u)) + n_b(u)\delta
<
n^*_r(\wright_r(v)) + n_b(v)\delta
$.
Proper intersection implies that these inequalities are strict,
that there is a path in the red dual \reddual~from $\wleft_r(v)$ to
$\wleft_r(u)$ to $\wright_r(v)$, and that the three faces are
distinct. This implies that $u$ belongs either to $B_r(v)$ or to
$T_r(v)$, and it lies in the corresponding regions of the plane, with
$\rface{u}$ (and hence the start/end of curves representing blue edges
incident to $u$) inside the region.
Similarly, by considering the blue dual \bluedual,   
$n^*_b(\wright_b(u)) + n_r(u)\delta
<
n^*_b(\bface{v}) + n_r(v)\delta
<
n^*_b(\wleft_b(u)) + n_r(u)\delta
$, we know that $v$ belongs either to $B_b(u)$, or to
$T_b(u)$, and it lies in the corresponding regions of the plane, with
$\bface{v}$ (and hence the start/end of curves representing red edges
incident to $v$) inside the region.
No matter which region, $B_r(v)$ or $T_r(v)$, vertex $u$ lies in, or
which region, $B_b(u)$ or $T_b(u)$, vertex $v$ lies in, the directed boundary
of the blue region ($B_b(u)$ or $T_b(u)$) containing $v$
crosses the directed boundary of the red region ($B_r(v)$ or $T_r(v)$)
containing $u$ from right to left.
This either violates condition {\bf c3} (if edges of the boundaries cross) or it
violates condition {\bf c2} (if the boundaries share a vertex). See Fig.~\ref{fi:intersectingLshapes} for an illustration.
\end{proof}

\begin{theorem}\label{th:algorithm}
Let $G_r$ and $G_b$ be two 
plane $st$-graphs defined on the same set of $n$ vertices $V$
and with distinct sources and sinks.
Let $\Phi:V \rightarrow \mathcal H=\{\SA,\SB,\SC,\SD\}$. If $\langle
G_r, G_b \rangle$ admits a \phit{}, then algorithm \algo~computes a 
\LSVR{} of $\langle G_r, G_b \rangle$ in $O(n)$ time. 
\end{theorem}
\begin{proof}
Lemmas~\ref{le:validstxy} and~\ref{le:lshapes} imply that
\algo~computes a \LSVR{} of $\langle G_r, G_b \rangle$. 
Computing the dual graphs and the four topological numberings ({\bf Steps 1-3}), as well as computing the two bar visibility representations and 
shifting each segment ({\bf Steps 4-5}), can be done in $O(n)$ time, as shown in~\cite{dett-gd-99,TamassiaTollis86}.
\end{proof}

\section{Testing Algorithm}\label{se:test}

In this section, we first show that there exists a pair of plane  $st$-graphs $\langle G_r,G_b \rangle$ 
that does not admit an \LSVR~for any possible function $\Phi$. 
This emphasizes the need for an efficient testing algorithm. 
Then, we show how to test whether two plane $st$-graphs with
the same set of  vertices admit a \LSVR{} for a given function $\phi$. 

\begin{theorem}\label{th:negative-fixedemb}
There exists a pair of 
plane  $st$-graphs $\langle G_r,G_b \rangle$ that does not admit a \LSVR~for any possible function $\Phi$.
\end{theorem}
\begin{proof}
Let $G_r$ be the plane $st$-graph drawn red in Fig.~\ref{fi:negative-1} (observe that it is a series-parallel graph,
i.e., a partial $2$-tree). Also, let $G_b$ be any plane $st$-graph containing the blue edge $(u,v)$ in Fig.~\ref{fi:negative-1}. 

\begin{figure}[h]
\centering
\includegraphics[page=1, scale=0.6]{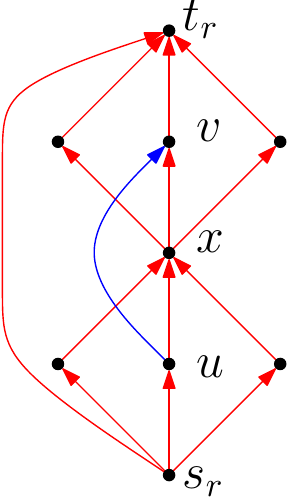}
\caption{
Illustration for the proof of Theorem~\ref{th:negative-fixedemb}. 
}\label{fi:negative-1}
\end{figure}

Edge $(u,v)$ is incident to either $\wright_r(u)$ or $\wleft_r(u)$ and to either $\wright_r(v)$ or $\wleft_r(v)$, based on $\Phi$. 
In any case edge $(u,v)$ will cross a red edge from left to right and a red edge from right to left, which implies that $\langle G_r,G_b \rangle$ 
does not admit any \phit, since condition {\bf c3} cannot be respected. 
It follows that, by Theorem~\ref{th:4polar}, $\langle G_r,G_b \rangle$ does not admit a \LSVR.
\end{proof}

Our testing algorithm exploits the interplay between the primal of the blue (red) graph and the dual of the red (blue) graph.
Given the circular order of the edges around each vertex imposed by the function $\phi$, we aim to compute a suitable path in the red dual for each blue edge. Such paths will then be used to route the blue edges. Finally, we check that no two blue edges cross.


We first introduce a few definitions. 
Let $G$ and $D^{\shortrightarrow}$ be a 
plane $st$-graph and its dual. Let $f$ and $g$ be two faces of $G$ that share an edge $e=(x,z)$ of $G$, such that $e$ belongs to the right (resp., left) path of $f$ (resp., $g$). Let $e^*$ be the dual edge in $D^{\shortrightarrow}$ corresponding to $e$. 
Let $w$ be a vertex on the right path of $f$ (or, equivalently, on the left path of $g$). Then $w$ is \emph{cut from  above} (resp., \emph{below}) by $e^*$, if $w$  precedes $z$ (resp., succeeds $x$) along the right path of $f$, i.e., all vertices that precede $z$ (including $x$) are cut from above, while all vertices that succeed $x$ (including $z$) are cut from below by $e^*$.

Let $G_r$ and $G_b$ be a pair of plane $st$-graphs with the same vertex set $V$
and with distinct sources and sinks. 
Let $\Phi:V \rightarrow \mathcal H=\{\SA,\SB,\SC,\SD\}$.
Recall that, for a given vertex $u$ of $G_{b}$, with the notation
$L_b(u)$, $R_b(u)$, $T_b(u)$ and $B_b(u)$ we represent the set of
vertices to the left, to the right, that are reachable from, and that
can reach $u$ in $G_b$, respectively (see
Section~\ref{se:preliminaries}).
Then consider an edge $e=(u,v)$ of $G_b$ and a path\footnote{Since \reddual is a multigraph, to uniquely
  identify $\pi_e$ we specify the edges that are traversed.
}  $\pi_e =
\{\rface{u} = f_0$, $e^*_{0}$, $f_1$, $\dots$, $f_{k-1}$, $e^*_{k-1}$,
$\rface{v}=f_k\}$ in \reddual, where $f_i$ ($0 \leq i \leq k$) are the faces traversed
by the path, and  $e^*_{i}$ ($0 \leq i < k$) are the dual edges
used by the path to go from $f_i$ to $f_{i+1}$.
Path $\pi_e$ is a \emph{traversing path} for $e$,
if $\pi_e =\{\rface{u} = \rface{v}\}$,  
or for all $0 \leq i < k$ and all vertices $w$ in the right path of
$f_i$: 
\begin{itemize}
\item[{\bf p1.}] If $w \in L_b(u)$ then
  $w$ is cut from below by $e^*_i$. See Fig.~\ref{fi:p1}.
\item[{\bf p2.}] If $w \in R_b(u)$ then
  $w$ is cut from above by $e^*_i$. See Fig.~\ref{fi:p2}.
\item[{\bf p3.}] If $w \in B_b(u)$ and $\Phi(w)=\SA$ (resp.,
  $\Phi(w)=\SD$) then $w$ is cut from above (resp., below) by
  $e^*_i$. See Fig.~\ref{fi:p3}.
\item[{\bf p4.}] If $w \in T_b(u)$ and $\Phi(w)=\SD$ (resp.,
  $\Phi(w)=\SA$) then $w$ is cut from above (resp., below) by
  $e^*_i$. See Fig.~\ref{fi:p4}.
\item[{\bf p5.}] If $w \in B_b(v)$ and $\Phi(w)=\SC$ (resp.,
  $\Phi(w)=\SB$) then $w$ is cut from above (resp., below) by
  $e^*_i$. See Fig.~\ref{fi:p5}.
\item[{\bf p6.}] If $w \in T_b(v)$ and $\Phi(w)=\SB$ (resp.,
  $\Phi(w)=\SC$) then $w$ is cut from above (resp., below) by
  $e^*_i$.  See Fig.~\ref{fi:p6}.
\end{itemize}  

\begin{figure}[t]
\centering
\subfigure[{\bf p1}]{\includegraphics[page=1]{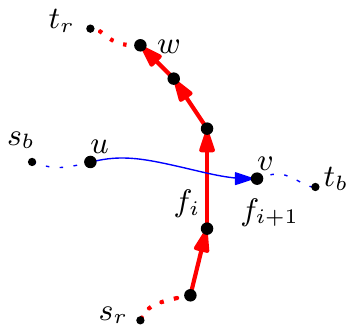}\label{fi:p1}}\hfil
\subfigure[{\bf p2}]{\includegraphics[page=2]{fig/travpath1}\label{fi:p2}}\hfil
\subfigure[{\bf p3}]{\includegraphics[page=3]{fig/travpath1}\label{fi:p3}}\hfil
\subfigure[{\bf p4}]{\includegraphics[page=4]{fig/travpath1}\label{fi:p4}}\hfil
\subfigure[{\bf p5}]{\includegraphics[page=5]{fig/travpath1}\label{fi:p5}}
\subfigure[{\bf p6}]{\includegraphics[page=6]{fig/travpath1}\label{fi:p6}}
\caption{Illustration for the properties of a traversing path $\pi_e$.}
\label{fi:travpath-properties}
\end{figure}

We now show that if $\langle G_r, G_b \rangle$ admits a \phit{}, then for each blue edge (the same argument would apply for red edges) there exists a unique traversing path. 

\begin{lemma}\label{le:uniquetp}
Let $G_r$ and $G_b$ be two 
plane $st$-graphs with the same vertex set $V$
and with distinct sources and sinks.
Let $\Phi:V \rightarrow \mathcal H=\{\SA,\SB,\SC,\SD\}$. If
$\langle G_r, G_b \rangle$
admits a \phit{}, then for every edge $e$ of $G_b$
there is a unique traversing path $\pi_e$ in \reddual.
\end{lemma}
\begin{proof}
If
$\langle G_r, G_b \rangle$
admits a \phit{} $G_{rb}$, then for every edge $e=(u,v)$ of $G_b$ there exists a path $\pi_e = \{\rface{u} = f_0$, $e^*_{0}$, $f_1$, $\dots$, $f_{k-1}$, $e^*_{k-1}$, $\rface{v}=f_k\}$ in \reddual, which is the path used by $e$ to go from $\rface{u}$ to $\rface{v}$ in $G_{rb}$. 

If $f_0$ and $f_k$ coincide, then $\pi_e$ is a traversing path. Otherwise, we would have a cycle $\pi_e=\{f_0=f_k,\dots,f_0=f_k\}$, which is not possible since \reddual~is acyclic, being the dual
of a plane $st$-graph.

If $f_0$ and $f_k$ do not coincide, let $w$ be
a vertex in the right path of $f_i$.
First, if $w$ belongs to $L_b(u)$, then it is cut from below.
Otherwise, if $w$ was cut
from above, since edge $e=(u,v)$ cannot cross the right path of $f_i$
twice (by condition {\bf c3}), it would belong to $R_b(u)$, a
contradiction with the fact that the embedding of $G_b$ is
preserved. Thus {\bf p1} is respected by $\pi_e$. With a symmetric
argument we can also prove {\bf p2}.
Suppose now that $w$  belongs to
$B_b(u)$, then  $\rface{w}=f_i=\wleft_r(w)$, otherwise if
$\rface{w}=f_{i+1}=\wright_r(w)$, the blue path from $w$ to $u$
would violate {\bf c3}. In other words, either $\Phi(w)=\SA$ or
$\Phi(w)=\SD$. Furthermore, if $\Phi(w)=\SA$, then $w$ must be cut
from above, while if $\Phi(w)=\SD$, then $w$ must be cut from below,
as otherwise the incoming blue edges to $w$ must enter a region
delimited by the blue path from $w$ to $u$, the blue edge $(u,v)$, and
part of the (red) right path of $f_i$, which violates the planarity of
the embedding of $G_b$ or condition {\bf c2}
(see Fig.~\ref{fi:p3}). Thus {\bf p3} is respected by $\pi_e$. With similar arguments
one can prove {\bf p4} -- {\bf p6}. Hence, $\pi_e$ is a traversing
set.
To prove that $\pi_e$ is unique, note that any possible traversing
set for $e$ must start from $f_0$ and leave this face. Hence, any vertex $w$ on
the right path of $f_0$ must be cut from either above or below,
according to properties  {\bf p1} -- {\bf p6} (which cover all
possible cases for $w$). The only edge that can satisfy the cut condition 
for all vertices on the right path of $f_0$, is an edge $e^*_0$ whose corresponding red primal edge,
denoted by $(x,z)$, is such that all vertices on the right path of
$f_0$ above $x$ must be cut
from below and all those below $z$ must be cut from above. Clearly, this edge is unique. By repeatedly applying this argument for
each face $f_i$ ($0 \leq i < k$), the traversing path $\pi_e$ is uniquely identified.
\end{proof}

The next theorem concludes the proof of Theorem~\ref{th:mainth}.

\begin{theorem}\label{th:4polar-test}
Let $G_r$ and $G_b$ be two 
plane $st$-graphs with 
the same set of $n$ vertices $V$
and with distinct sources and sinks.
Let $\Phi:V \rightarrow \mathcal H=\{\SA,\SB,\SC,\SD\}$. There exists
an $O(n^3)$-time algorithm to test whether $\langle G_r, G_b \rangle$
admits a \phit{}.
\end{theorem}
\begin{proof}
Our testing algorithm aims to compute (if it exists) a \phit{} $G_{rb}$ for
$\langle G_r, G_b \rangle$.
We first fix the circular order of the edges restricted to the blue
edges (resp., red edges) around each vertex $u$ of $G_{rb}$ to satisfy
{\bf c1} and to maintain the planar embedding of $G_b$ (resp.,
$G_r$). We then fix the circular order of the blue edges with respect
to the red edges around each vertex $u$ of $G_{rb}$ to satisfy {\bf
  c2} (i.e., to obey $\Phi(u)$). Then, we first check if for every
blue edge $e$ there exists a traversing path $\pi_e$; if so, we verify
that by routing every blue edge $e$ through $\pi_e$ no two blue edges
cross each other. If this procedure succeeds then
$\langle G_r, G_b \rangle$
admits \phit{} $G_{rb}$, because, by construction, the resulting
embedding of $G_{rb}$ satisfies conditions {\bf c1}, {\bf c2} and {\bf
  c3}. Otherwise, either there exists a blue edge with no traversing
path, or two traversing paths are such that the two corresponding
edges of $G_b$ cross if routed through them. In the first case
$\langle G_r, G_b \rangle$
does not admit a \phit{} by Lemma~\ref{le:uniquetp}. In the second case,
since the traversing paths are unique, condition {\bf c2} cannot be
satisfied, and again $\langle G_r, G_b \rangle$ does not admit a \phit{}. 

The testing algorithm works in two phases as follows.

{\bf Phase 1.} For every edge $e=(u,v) \in E_b$. If $\rface{u} =
\rface{v}$, we have found a traversing path. Otherwise, we label each
vertex on the right path of $\rface{u}$, by $A$ if it must be cut from
above or by $B$ if it must be cut from below, according to properties
{\bf p1} -- {\bf p6}. Then we check if the sequence of labels along
the path is a nonzero number of $A$'s followed by a nonzero number of
$B$'s.
If so, then the dual
edge of the traversing path is the one whose corresponding primal edge
has the two end-vertices with different labels (which is unique). If
this is not the case, then a traversing path for $e$ does not exist. In
the positive case, we add the dual edge we found and the next face we
reach through this edge to $\pi_e$ and we iterate the algorithm until
we reach either $\rface{v}$ or the outer face of \reddual. In the
former case $\pi_e$ is a traversing path for $e$, while in the latter
case, since the edges of the outer face of $G_{rb}$ cannot be crossed by
definition of \phit{}, we have that again no traversing path can be found.

{\bf Phase 2.} We now check that by routing every edge $e \in E_b$
through its corresponding traversing path $\pi_e$, no two of these
edges cross each other.
Consider the dual graph \reddual, which is a 
plane $st$-graph. Construct a planar drawing $\Gamma$ of \reddual. Consider
any two traversing paths $\pi_e$ and $\pi_e'$, which corresponds to
two paths in $\Gamma$, and let $e=(u,v)$ and $e'=(w,z)$ be the two
corresponding  edges of $G_b$. Denote by $\hat \pi_e = \{u\} \cup \pi_e \cup \{v\}$
and $\hat \pi_e' = \{w\} \cup \pi_e' \cup \{z\}$ the two enriched
paths. Enrich $\Gamma$ by adding the four edges $(x,\rface{x})$, where $x \in \{u,v,w,z\}$, 
in a planar way  respecting the original embedding of $G_b$. 
Consider now the subdrawing $\Gamma'$ of $\Gamma$
induced by $\hat \pi_e \cup \hat \pi_e'$. See Fig.~\ref{fi:crossingpaths} for an illustration. 
If $e$ and $e'$ cross each other, then $\pi_e \cap \pi_e'$ cannot be
empty. Moreover, the intersection $\pi_e \cap \pi_e'$ must be a single
subpath, as otherwise the two traversing paths would
not be unique. Let $f$ be the first face and let $g$ be the last face in this subpath. 
Let $e_u$ be the incoming edge of $f$ that belongs to the subpath of $\hat \pi_e$ from $u$ to $f$; 
and let $e_w$ be the incoming edge of $f$ that belongs to the subpath of $\hat \pi_e'$ from $w$ to $f$. 
Also, let $e_v$ be the outgoing edge of $g$ that belongs to the subpath of $\hat \pi_e$ from $g$ to $v$; 
and let $e_z$ be the outgoing edge of $g$ that belongs to the subpath of $\hat \pi_e'$ from $g$ to $z$. 
Then $e$ and $e'$ cross if and only if walking clockwise along $\pi_e \cup \pi_e'$ from $f$ to $g$ and back to $f$ these four edges 
are encountered in the circular order $e_u$, $e_z$, $e_v$, $e_w$ as shown in Fig.~\ref{fi:crossingpaths}. 
Note that, $e_u$ and $e_w$ may coincide if $u=w$, and similarly for $e_v$ and $e_z$.

\begin{figure}[h]
\centering
\includegraphics[scale=0.6]{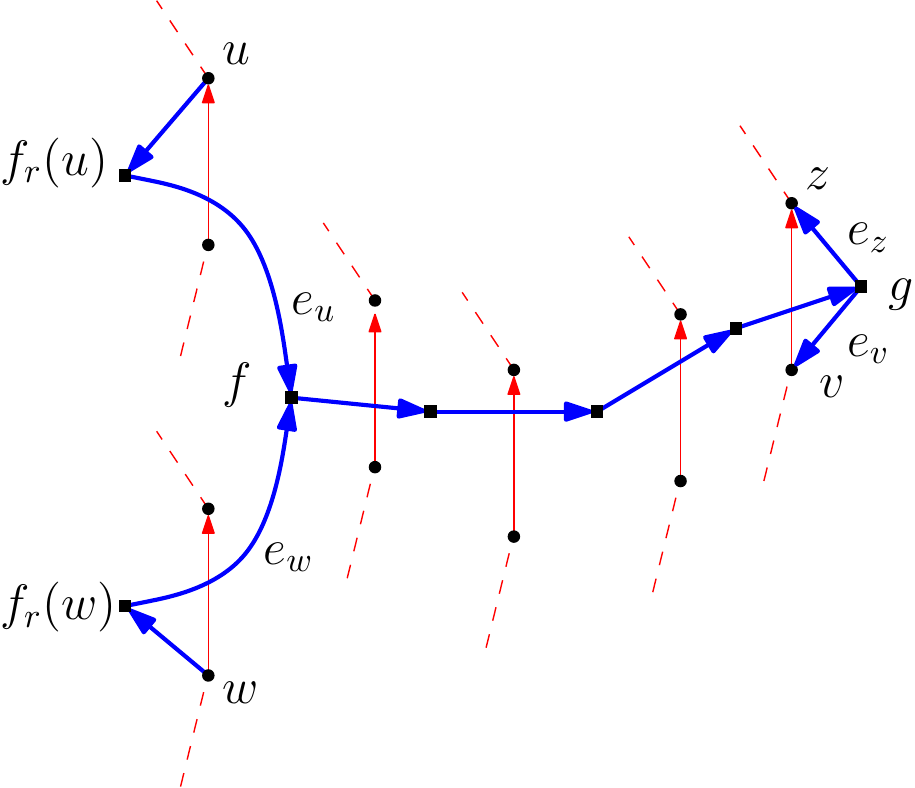}
\caption{Illustration for the proof of Theorem~\ref{th:4polar-test}.}\label{fi:crossingpaths}
\end{figure}

To conclude the proof we need to show that {\bf Phase 1} and {\bf Phase 2} can be implemented in $O(n^3)$ time. 
Checking if a traversing path exists for each edge in $E_{b}$ ({\bf
  Phase 1}) can be performed in $O(n^3)$ time. Namely, for each  of
these edges, we need to assign a label $A$ or $B$ to every vertex $w$
on the right path of the traversed faces. This can be done by checking
the existence of a path in the primal graph $G_b$ or in the dual graph
\bluedual, as by Lemma~\ref{le:tt}, thus requiring $O(n)$ time per
vertex, $O(n^2)$ time per edge, and $O(n^3)$ time in
total. Furthermore, checking if two of these edges would cross if
routed through their traversing paths ({\bf Phase 2}) can also be
performed in $O(n^3)$ time. In fact, constructing a planar drawing of
\reddual~costs $O(n)$ time~\cite{cp-ltadp-IPL95}, while checking if
two paths cross in the drawing costs $O(n)$ time for each pair of
paths, of which there are $O(n^2)$ in total. This concludes the proof.
\end{proof}

\section{Simultaneous RAC Embeddings via Simultaneous Visibility Representations with L-shapes}\label{se:wheel+matching}

In this section, we show that Theorem~\ref{th:mainth} can be used to shed more light on the problem of
computing a simultaneous RAC embedding~\cite{DBLP:journals/jgaa/ArgyriouBKS13,DBLP:conf/walcom/BekosDKW15}. 

Given two planar graphs with the same vertex set, an SRE is a
simultaneous embedding where crossings between edges of the two graphs occur at right angles. 
Argyriou {\em et al.} proved that it is always possible to construct an SRE with straight-line edges 
of a cycle and a matching, while there exist a wheel graph and a cycle that do not admit such a representation~\cite{DBLP:journals/jgaa/ArgyriouBKS13}.  
This motivated recent results about SRE with bends along the edges. 
Namely, Bekos {\em et al.} show that two planar graphs with the same vertex set admit an SRE with at most six bends per edge in both graphs~\cite{DBLP:conf/walcom/BekosDKW15}.

We observe that any pair of graphs that admit a simultaneous
visibility representation with L-shapes also admits an SRE 
with at most two bends per edge. This is obtained 
with the technique used in Section~\ref{se:characterization} to compute a \phit~from a \LSVR, see Fig.~\ref{fi:necessity}. 
Thus, a new approach to characterize graph
pairs that have SREs with at most two bends per
edge is as follows: Given two planar graphs with the same vertex
set, add to each of them a unique source and a unique sink, and look for two $st$-orientations (one for each of the two graphs) and a function $\Phi$ such
that the two graphs admit a \LSVR{}.

As an application of this result, in what follows we show  an alternative proof of another result by Bekos {\em et al.}
that a wheel graph and a matching admit an SRE with at most two bends for each edge 
of the wheel, and no bends for the matching edges~\cite{DBLP:conf/walcom/BekosDKW15}.

\medskip

\begin{figure}[t]
\centering \includegraphics{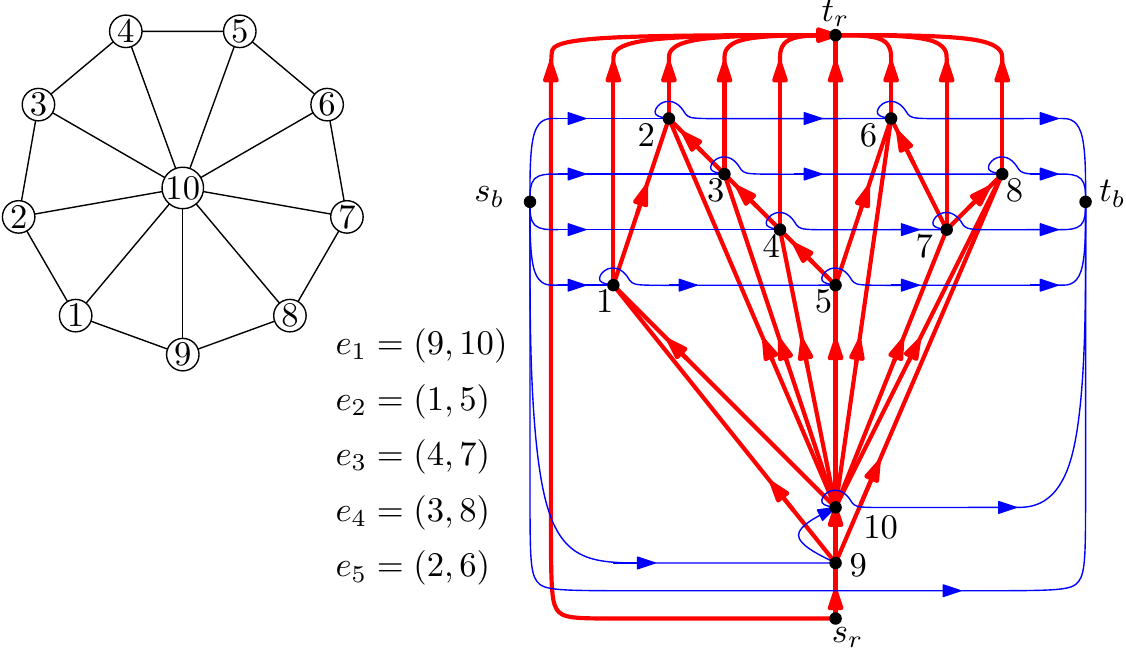}
\caption{The wheel $W_{10}$, edges of a matching $M$, and the \phit~representing both.}
\label{fi:W10M}
\end{figure}

The \emph{wheel} on $n$ vertices, denoted by $W_n$, is an undirected
graph with vertices $V=\{1,2,\dots,n\}$ and edges $\{(i,n) : 1 \leq
i \leq n-1 \}$, called \emph{spokes}; and $\{ (i,i+1) : 1 \leq i
\leq n-2 \} \cup \{(1,n-1)\}$, called \emph{rim edges}. See
Fig.~\ref{fi:W10M} for an illustration. A graph is a \emph{matching} if no two of its edges share a common
end-vertex and a \emph{complete matching} if every vertex is the
end-vertex of some edge.

\medskip

We begin with the following observation, based on Theorem~\ref{th:4polar}.

\newtheorem{observation}{Observation}

\begin{observation}
Let $\langle G_r'=(V,E_r'), G_b'=(V,E_b') \rangle$  be a pair of undirected plane graphs. 
Let $\Phi:V \rightarrow \mathcal H=\{\SA\}$.  
Then $\langle G_r', G_b' \rangle$ admits a \LSVR, if and only if
$G_r'$ and $G_b'$ can be augmented and $st$-oriented to two plane
$st$-graphs $G_r$ and $G_b$, respectively, defined on the same set of
vertices and with distinct sources and sinks, such that $\langle G_r,
G_b \rangle$ admits a \phit.
\end{observation}

\medskip

We show that the wheel $W_n$ and any
matching $M$ on the vertices of $W_n$ has a \LSVR~ where $\Phi(v) =
\SA$ for all vertices $v$ in $W_n$.
To do this,
we construct a \phit~ for the pair $\langle G_r, G_b
\rangle$, where $G_r$ and $G_b$ are plane $st$-graphs obtained from $W_n$
and $M$, respectively.

We assume that $n$ is even and $M$ is a complete matching, though we
will later see how to handle partial matchings.
By symmetry, we may assume that vertex $n$ is matched with
vertex $n-1$.
Let $e_1,e_2, \dots, e_{n/2}$ be the edges in $M$ where $e_1 = (n-1,n)$.
Let the \emph{height}, $h(u)$, of a vertex $u$ be the index of its
edge in the matching, so if $u \in e_i$ then $h(u) = i$.

For all $u \in \{1,2,\dots,n-2\}$, use the point $(u, h(u))$ to
represent vertex $u$.
Use the point $(n/2, -n^2/2)$ to represent vertex $n$.
Use the point $(n/2, -n^2/2-1)$ to represent vertex $n-1$.
The drawing of $W_n$ is straight-line planar because
the spokes of the wheel do not intersect the rim edges since the slope
of any rim edge (excluding $(n-2,n-1)$ and $(1,n-1)$)
is in the range $(-n/2,+n/2)$. 

To define $G_r$ (the $st$-graph representing $W_n$), direct every rim
edge $(u,v)$ from smaller to larger height vertex.
If both vertices have the same height (i.e., the edge is part of the
matching), direct the edge from smaller to larger
numbered vertex.
Direct spoke edges $(n,u)$ for $1 \leq u \leq n-2$ from vertex $n$ to
vertex $u$.
Direct spoke edge $(n-1,n)$ from vertex $n-1$ to vertex $n$.
Add the directed edge $(s_r,n-1)$ and the directed edges $(u,t_r)$ for all
$1 \leq u \leq n-1$, where the points $(n/2,-\infty)$ and
$(n/2,+\infty)$ represent $s_r$ and $t_r$ respectively.
Finally, add the edge $(s_r, t_r)$ on the outer face of $G_r$.

The drawing of $M$ is also straight-line planar since its edges are
all horizontal segments with different positive $y$-coordinates,
except for $(n-1,n)$, which is vertical and spans only negative
$y$-coordinates.
To define $G_b$ (the $st$-graph representing $M$),
direct all edges from
smaller to larger numbered end-point, so the edges are directed from
left to right.
Add the directed edge $(s_b,u)$ (resp., $(u,t_b)$) for each vertex $u$
that is the smaller (resp., larger) numbered end-point of an
edge in $M$, where the points $(-\infty,0)$ and
$(+\infty,0)$ represent $s_b$ and $t_b$, respectively.
Finally, add the edge $(s_b, t_b)$ on the outer face of $G_b$.

To ensure condition {\bf c2} and preserve condition {\bf c3} of
Theorem~\ref{th:4polar}, we slightly
modify the drawing of the edges in $G_b$ (i.e., the blue edges).
When drawn as straight-line segments, the head of every blue edge
terminates in the left red face of its destination as desired, but
the tail does not leave from the left red face of its origin, in general.
This is easily fixed by drawing each blue edge (except those incident
to $s_b$) so that it
leaves from the origin's left red face, heading to the left and slightly
above the entering blue edge, and then loops to the right over the
vertex to rejoin its original rightward horizontal path.  See
Fig.~\ref{fi:W10M}.

All red edges (from $W_n$) are directed from bottom to top.
Blue edges (from $M$), except for $(n-1,n)$ which is identical to the
corresponding red edge (from $W_n$), intersect red edges only during
the blue edges' rightward trajectories.
Thus blue edges only cross red edges from left to right.

To represent a partial matching $P$, which is a complete matching $M$ without
some edges, for all edges $(u,v) \in M \setminus P$, we
add edge $(s_b,v)$ and $(u,t_b)$ to $G_b$, routing the first slightly
above and the second slightly below the existing blue path $\langle s_b, u, v,
t_b \rangle$,
and remove the edge $(u,v)$ from $G_b$.

\section{Final Remarks and Open Problems}\label{se:discussion}

In this paper we have introduced and studied the concept of
simultaneous visibility representation with L-shapes of two 
plane $st$-graphs. 
We remark that it is possible to include in our 
theory the case when the vertices can also be drawn
as rectangles. 
Nevertheless, this would not enlarge the class of representable pairs of graphs. 
In fact, for every vertex $v$ drawn as a rectangle $\mathcal R_v$, we can replace $\mathcal R_v$ with any L-shape by keeping only two adjacent sides of $\mathcal R_v$ in the drawing and prolonging the visibilites incident to the removed sides of  $\mathcal R_v$. 
The converse is not true. Indeed, roughly speaking, L-shapes can be nested, whereas rectangles cannot. To give an example, if a vertex $v$ must see a vertex $u$ both vertically and horizontally, this immediately implies that the two corresponding rectangles need to overlap, while two L-shapes could instead be nested. 
Several extensions of the model introduced in
this paper can also be studied, e.g., the case
where every edge is represented by a T-shape, or more generally by a
$+$-shape.

Three questions that stem from this
paper are
whether the time complexity of the testing algorithm in
Section~\ref{se:test} can be improved;
what is the complexity of deciding
if two given plane $st$-graphs admit a \LSVR{} for some function
$\Phi$, which is not part of the input;
and
what is the complexity of deciding
if two undirected graphs admit a \LSVR{} for some function
$\Phi$.

\bibliography{simVG}

\bibliographystyle{abbrv}

\end{document}